\documentclass[article,shortnames,nojss]{jss}

\usepackage{graphicx}
\usepackage{amsfonts}
\usepackage{amssymb}
\usepackage{amsmath}
\usepackage{amsthm}
\usepackage{multirow}
\usepackage[english]{babel}

\newcommand{\given}{\operatorname{|}}
\newcommand{\indep}{\perp\hspace{-0.21cm}\perp}
\newcommand{\tr}{\operatorname{tr}}

\newtheorem{thm}{Theorem}
\newtheorem{ex}{Example}

\author{Marco Scutari\\University of Padova}
\title{Structure Variability in Bayesian Networks}
\Plainauthor{Marco Scutari}
\Plaintitle{Structure Variability in Bayesian Networks}

\Abstract{

The structure of a Bayesian network encodes most of the information about the
probability distribution of the data, which is uniquely identified given some
general distributional assumptions. Therefore it's important to study the
variability of its network structure, which can be used to compare the 
performance of different learning algorithms and to measure the strength of
any arbitrary subset of arcs.

In this paper we will introduce some descriptive statistics and the
corresponding parametric and Monte Carlo tests on the undirected graph
underlying the structure of a Bayesian network, modeled as a multivariate 
Bernoulli random variable.

}
\Keywords{Bayesian network, bootstrap, multivariate Bernoulli distribution, 
  structure learning algorithms}
\Plainkeywords{Bayesian network, bootstrap, multivariate Bernoulli
  distribution, structure learning algorithms}

\Address{
  Marco Scutari\\
  Department of Statistical Sciences\\
  University of Padova\\
  Via Cesare Battisti 241, 35121 Padova, Italy\\
  E-mail: \email{marco.scutari@stat.unipd.it}
}

\begin{document}

\section{Introduction}

In recent years Bayesian networks have been successfully applied in several
different disciplines, including medicine, biology and epidemiology (see
for example \citet{gene} and \citet{holmes}). This has been made possible by a
rapid evolution of structure learning algorithms, both constraint-based (from
PC \citep{spirtes} to Grow-Shrink \citep{mphd} to IAMB \citep{iamb} and its
variants \citep{fastiamb}) and score-based (from Greedy Equivalent Search
\citep{ges} to genetic algorithms \citep{larranaga}).

The main goal in the development of these algorithms was the reduction of
the number of either independence tests or score comparisons needed to learn
the structure of the Bayesian network. Their correctness was proved assuming
either very large sample sizes in relation to the number of variables (in the
case of Greed Equivalent Search) or the absence of both false positives and
false negatives (in the case of Grow-Shrink and IAMB). In most cases the
characteristics of the learned networks were studied using a small number of
reference data sets \citep{bnr} as benchmarks, and differences from the true
structure measured with descriptive measures such as Hamming distance
\citep{graphs}.

This approach to model evaluation is not possible for real world data sets, as
the true structure of their probability distribution is not known in advance.
An alternative is provided by the use of either parametric or nonparametric
bootstrap \citep{efron}. By applying the learning algorithm to a sufficiently
large number of bootstrap samples it is possible to obtain confidence
intervals and empirical probabilities for any feature of the network structure
\citep{friedman}, such as the presence or the composition of the Markov Blanket
of a particular node. The fundamental limit in the interpretation of the
results is that the ``reasonable'' level of confidence for thresholding depends
on the data.

In this paper we propose a modified bootstrap-based approach for the inference
on the structure of a Bayesian network. The undirected graph underlying the
network structure is modeled as a multivariate Bernoulli random variable in
which each component is associated with an arc. This assumption allows the
derivation of both exact and asymptotic measures of the variability of the
network structure or its parts.

\section{Bayesian networks and bootstrap}
\label{sec:bn}

Bayesian networks are graphical models where nodes represent random variables
(the two terms are used interchangeably in this article) and arcs represent
probabilistic dependencies between them \citep{korb}.

The graphical structure $\mathcal{G} = (\mathbf{V}, A)$ of a Bayesian network is a
\textit{directed acyclic graph} (DAG) which defines a factorization of the
joint probability distribution of $\mathbf{V} = \{X_1, X_2, \ldots, X_v\}$,
often called the \textit{global probability distribution}, into a set of
\textit{local probability distributions}, one for each variable.
The form of the factorization is given by the \textit{Markov property},
which states that every random variable $X_i$ directly depends only on
its parents $\Pi_{X_i}$:
\begin{align}
  &\Prob(X_1, \ldots, X_v) = \prod_{i=1}^v \Prob(X_i \given \Pi_{X_i})&
  & \text{(for discrete variables)}\\
  &f(X_1, \ldots, X_v) = \prod_{i=1}^v f(X_i \given \Pi_{X_i})&
  & \text{(for continuous variables).}
\end{align}

Therefore it's important to define confidence measures for specific features
in the network structure, such as the presence of specific configurations of
arcs. A related problem is the definition of a measure of variability for
the network structure as a whole, both as an indicator of goodness of fit for
a particular Bayesian network and as a criterion to evaluate the performance
of a learning algorithm.

A possible solution for both these problems has been developed by \citet{friedman} using
bootstrap simulation, and modified by \citet{imoto} to estimate the confidence
in the presence of an arc (called \textit{edge intensity}, and also known as
\textit{arc strength}) and its direction.
This approach can be summarized as follows:
\begin{enumerate}
  \item For $b = 1, 2, \ldots, m$
    \begin{enumerate}
      \item re-sample a new data set $\mathbf{D^*_b}$ from the original data
        $\mathbf{D}$ using either parametric or nonparametric bootstrap.
      \item learn a Bayesian network $\mathcal{G}_b$ from $\mathbf{D^*_b}$.
    \end{enumerate}
  \item Estimate the confidence in each feature $f$ of interest as
    \begin{equation}
       \hat\Prob(f) = \frac{1}{m} \sum_{b=1}^m f(\mathcal{G}_b).
    \end{equation}
\end{enumerate}

However, the empirical probabilities $\hat\Prob(f)$ are difficult to evaluate,
because the distribution of $\mathcal{G}$ is unknown and the confidence 
threshold value depends on the data.

\section{The multivariate Bernoulli distribution}
\label{sec:mvb}

Let $B_1, B_2, \ldots, B_k$, $k \in \mathbb{N}$ be Bernoulli random variables
with marginal probability of success $p_1, p_2, \ldots, p_k$, that is
$B_i \sim Ber(p_i)$, $i = 1, \ldots, k$. Then the distribution of the random
vector $\mathbf{B} = [B_1, B_2, \ldots, B_k]^T$ over the joint probability
space of $B_1, B_2, \ldots, B_k$ is a \textit{multivariate Bernoulli random
variable} \citep{krummenauer}, denoted as $Ber_k(\mathbf{p})$. Its
probability function is uniquely identified by the parameter collection
\begin{equation}
  \mathbf{p} = \left\{ p_I : I \subseteq \{1, \ldots, k \}, I \neq \varnothing \right\},
\end{equation}
which represents the \textit{dependence structure} among the marginal distributions
in terms of simultaneous successes for every non-empty subset $I$ of elements of the
random vector.

However, several useful results depend only on the first and second order moments
of $\mathbf{B}$
\begin{align}
  \E(B_i) &= p_i \\
  \label{eq:var} \VAR(B_i) &= \E(B_i^2) - \E(B_i)^2 = p_i - p_i^2 \\
  \label{eq:cov} \COV(B_i, B_j) &= \E(B_i B_j) - \E(B_i) \E(B_j)= p_{ij} - p_i p_j
\end{align}
and the reduced parameter collection
\begin{equation}
  \mathbf{\tilde{p}} = \left\{ p_{ij} : i,j = 1, \ldots, k \right\},
\end{equation}
which is in fact used as an approximation of $\mathbf{p}$ in the generation
random multivariate Bernoulli vectors in \citet{mvbsim}.

\subsection{Uncorrelation and independence}

Let's first consider a simple result that links covariance (and therefore correlation)
and independence of two univariate Bernoulli variables.
\begin{thm}
\label{thm:univindep}
  Let $B_i$ and $B_j$ be two Bernoulli random variables. Then $B_i$ and $B_j$
  are independent if and only if their covariance is zero:
  \begin{equation}
    B_i \indep B_j \Longleftrightarrow \COV(B_i, B_j) = 0
  \end{equation}
\end{thm}
\begin{proof}
  If $B_i$ and $B_j$ are independent then by definition
  \begin{equation*}
    \COV(B_i, B_j) = p_{ij} - p_i p_j
      = \Prob(B_i = 1,B_j = 1) - \Prob(B_i = 1)\Prob(B_j = 1) = 0,
  \end{equation*}
  as $\Prob(B_i = 1, B_j = 1) = \Prob(B_i = 1)\Prob(B_j = 1)$.

  If on the other hand we have that $\COV(B_i, B_j) = 0$, then
  \begin{equation*}
    p_{ij} - p_i p_j = 0 \Rightarrow p_{ij} = p_i p_j \Rightarrow B_i \indep B_j
  \end{equation*}
  which completes the proof.
\end{proof}

This theorem can be extended to multivariate Bernoulli random variables as follows.

\begin{thm}
  Let $\mathbf{B} = [B_1, B_2, \ldots, B_k]^T$ and $\mathbf{C} = [C_1, C_2, \ldots, C_l]^T$,
  $k,l \in \mathbb{N}$ be two multivariate Bernoulli random variables. Then $\mathbf{B}$ and
  $\mathbf{C}$ are independent if and only if
  \begin{equation}
    \mathbf{B} \indep \mathbf{C} \Longleftrightarrow \COV(\mathbf{B}, \mathbf{C}) = \mathbf{O}
  \end{equation}
  where $\mathbf{O}$ is the zero matrix.
\end{thm}
\begin{proof}
  If $\mathbf{B}$ is independent from $\mathbf{C}$, then by definition every pair
  $(B_i, C_j)$, $i = 1, \ldots, k$, $j = 1, \ldots, l$ is independent. Therefore the
  covariance matrix of $\mathbf{B}$ and $\mathbf{C}$ is
  \begin{equation*}
    \COV(B_i, C_j) = c_{ij} = 0
    \Longrightarrow
    \COV(\mathbf{B}, \mathbf{C}) = [c_{ij}] = \mathbf{O}
  \end{equation*}
  If conversely the covariance matrix $\COV(\mathbf{B}, \mathbf{C})$ is equal to the zero
  matrix, every pair $(B_i, C_j)$ is independent as
  \begin{equation*}
    c_{ij} = p_{ij} - p_i p_j = 0 \Longrightarrow p_{ij} = p_i p_j
  \end{equation*}
  This implies the independence of the random vectors $\mathbf{B}$ and $\mathbf{C}$, as
  their sigma-algebras
  \begin{align*}
    &\sigma(\mathbf{B}) = \sigma(B_1) \times \ldots \times \sigma(B_k)&
    &\text{and}&
    &\sigma(\mathbf{C}) = \sigma(C_1) \times \ldots \times \sigma(C_l)
  \end{align*}
  are functions of the sigma algebras induced by the two sets of independent random
  variables $B_1, B_2, \ldots, B_k$ and $C_1, C_2, \ldots, C_l$.
\end{proof}

The correspondence between uncorrelation and independence is identical to the
analogous property of the multivariate Gaussian distribution \citep{ash}, and
is closely related to the strong normality defined for orthogonal second order
random variables in \citet{loeve}. It can also be applied to disjoint subsets
of components of a single multivariate Bernoulli variable, as they are also
distributed as multivariate Bernoulli random variables.

\begin{thm}
  Let $\mathbf{B} = [B_1, B_2, \ldots, B_k]^T$ be a multivariate Bernoulli random
  variable; then every random vector $\mathbf{B^*} = [B_{i_1}, B_{i_2}, \ldots,
  B_{i_l}]^T$, $\left\{i_1, i_2, \ldots, i_l\right\} \subseteq \left\{1, 2, \ldots,
  k\right\}$ is a multivariate Bernoulli random variable.
\end{thm}
\begin{proof}
  The marginal components of $\mathbf{B^*}$ are Bernoulli random variables, because
  $\mathbf{B}$ is multivariate Bernoulli. The new dependency structure is defined as
  \begin{equation*}
    \mathbf{p^*} = \left\{ p_{I^*} : I^* \subseteq \{i_1, \ldots, i_l \} \subseteq \{1, \ldots, k \}, I^* \neq \varnothing \right\},
  \end{equation*}
  and uniquely identifies the probability distribution of $\mathbf{B^*}$.
\end{proof}

\begin{ex}
  Let's consider the trivariate Bernoulli random variable
  \begin{align*}
  &\mathbf{B} = \begin{bmatrix} B_1 \\ B_2 \\ B_3 \end{bmatrix} = \mathbf{B_1} + \mathbf{B_2}&
  &\text{where}& &\mathbf{B_1} = \begin{bmatrix} 0 \\ B_2 \\ 0 \end{bmatrix}&
  &\text{and}& &\mathbf{B_2} = \begin{bmatrix} B_1 \\ 0 \\ B_3 \end{bmatrix}.
  \end{align*}
  Then the covariance matrix
  \begin{align*}
    \COV(\mathbf{B_1}, \mathbf{B_2}) &=
      \E\left(\begin{bmatrix} 0 \\ B_2 \\ 0 \end{bmatrix}
              \begin{bmatrix} B_1 & 0 & B_3 \end{bmatrix} \right) -
      \E\left(\begin{bmatrix} 0 \\ B_2 \\ 0 \end{bmatrix}\right)
      \E\left(\begin{bmatrix} B_1 & 0 & B_3 \end{bmatrix} \right) \\
      &= \E\left(\begin{bmatrix}
                   0       & 0 & 0       \\
                   B_1 B_2 & 0 & B_2 B_3 \\
                   0       & 0 & 0       \\
                 \end{bmatrix} \right) -
      \begin{bmatrix} 0 \\ p_2 \\ 0 \end{bmatrix}
      \begin{bmatrix} p_1 & 0 & p_3 \end{bmatrix} \\
      &= \begin{bmatrix}
           0       & 0 & 0      \\
           p_{12}  & 0 & p_{23} \\
           0       & 0 & 0      \\
         \end{bmatrix} -
         \begin{bmatrix}
           0        & 0 & 0       \\
           p_1 p_2  & 0 & p_2 p_3 \\
           0        & 0 & 0       \\
         \end{bmatrix} = \\
      &= \begin{bmatrix}
           0       & 0 & 0      \\
           p_{12} - p_1 p_2 & 0 & p_{23} - p_2 p_3\\
           0       & 0 & 0      \\
         \end{bmatrix}
  \end{align*}
  of the two components $\mathbf{B_1}$ and $\mathbf{B_2}$ is equal to the zero
  matrix if and only if
  \begin{align*}
    \left\{
    \begin{aligned}
     p_{12} &= p_1 p_2 \\
     p_{23} &= p_2 p_3
    \end{aligned}
    \right. \Longrightarrow \left\{ B_1 \indep B_2, B_2 \indep B_3\right\}
  \end{align*}
  which in turn implies and is implied by $\mathbf{B_1} \indep \mathbf{B_2}$.
\end{ex}

\subsection{Properties of the covariance matrix}

The covariance matrix $\Sigma = [\sigma_{ij}]$, $i,j = 1, \ldots, k$ associated
with a multivariate Bernoulli random vector has several interesting numerical
properties. Due to the form of the central second order moments defined in
formulas \ref{eq:var} and \ref{eq:cov}, the diagonal elements are bound in the
interval
\begin{equation}
  \label{thm:varconstr}
  \sigma_{ii} = p_i - p_i^2 \in \left[0, \frac{1}{4}\right].
\end{equation}
The maximum is attained for $p_i = \frac{1}{2}$, and the minimum for both
$p_i = 0$ and $p_i = 1$. For the Cauchy-Schwartz theorem \citep{ash} then
\begin{equation}
  0 \leqslant \sigma_{ij}^2 \leqslant \sigma_{ii} \sigma_{jj} \leqslant \frac{1}{16}
  \Longrightarrow
  |\sigma_{ij}| \in \left[0, \frac{1}{4}\right].
\end{equation}

The eigenvalues $\lambda_1, \lambda_2, \ldots, \lambda_k$ of $\Sigma$ are
similarly bounded, as shown in the following theorem.

\begin{thm}
  Let $\mathbf{B} = [B_1, B_2, \ldots, B_k]^T$ be a multivariate Bernoulli random
  variable, and let $\Sigma = [\sigma_{ij}]$, $i,j = 1, \ldots, k$ be its covariance
  matrix. Let $\lambda_i$, $i = 1, \ldots, k$ be the eigenvalues of $\Sigma$. Then
  \begin{equation}
    0 \leqslant \sum_{i=1}^k \lambda_i \leqslant \frac{k}{4}
  \end{equation}
  and
  \begin{equation}
    0 \leqslant \lambda_i \leqslant \frac{k}{4}.
  \end{equation}
\end{thm}
\begin{proof}
  Since $\Sigma$ is a real, symmetric, non-negative definite matrix, the eigenvalues
  $\lambda_i$ are non-negative real numbers \citep{salce}; this proves the lower bound
  in both inequalities.

  The upper bound in the first inequality holds because
  \begin{equation*}
     \sum_{i=1}^k \lambda_i = \sum_{i=1}^k \sigma_{ii} \leqslant
     \max_{\left\{\sigma_{ii}\right\}} \sum_{i=1}^k \sigma_{ii} =
     \sum_{i=1}^k \max \sigma_{ii} = \frac{k}{4},
  \end{equation*}
  as the sum of the eigenvalues is equal to the trace of $\Sigma$ \citep{seber}. This
  in turn implies
  \begin{equation*}
    \lambda_i \leqslant \sum_{i=1}^k \lambda_i \leqslant \frac{k}{4},
  \end{equation*}
  which completes the proof.
\end{proof}

These bounds define a convex set in $\mathbb{R}^k$, defined by the family
\begin{equation}
  \mathcal{D} = \left\{ \Delta^{k-1}(c) : c \in \left[ 0, \frac{k}{4} \right]\right\}
\end{equation}
where $\Delta^{k-1}(c)$ is the non-standard $k-1$ simplex
\begin{equation}
  \Delta^{k-1}(c) = \left\{ (\lambda_1, \ldots, \lambda_k) \in \mathbb{R}^k :
  \sum_{i=1}^k \lambda_i = c, \lambda_i \geqslant 0\right\}.
\end{equation}

\subsection{Sequences of multivariate Bernoulli variables}

Let's now consider a sequence of independent and identically distributed
multivariate Bernoulli variables $\mathbf{B_1}, \mathbf{B_2}, \ldots, \mathbf{B_m}
\sim Ber_k(\mathbf{p})$.
The sum
\begin{equation}
  \mathbf{S}_m = \sum_{i=1}^m \mathbf{B}_i \sim Bi_k(m, \mathbf{p})
\end{equation}
is distributed as a \textit{multivariate Binomial random variable} \citep{krummenauer},
thus preserving one of the fundamental properties of the univariate Bernoulli
distribution. A similar result holds for the \textit{law of small numbers},
whose multivariate version states that a $k$-variate Binomial distribution
$Bi_k(m, \mathbf{p})$ converges to a \textit{multivariate Poisson distribution}
$P_k(\mathbf{\Lambda})$:
\begin{align}
  &\mathbf{S}_m \stackrel{d}{\to} P_k(\mathbf{\Lambda})&
  &\text{as}&
  &m\mathbf{p} \to \mathbf{\Lambda}.
\end{align}

Both these distributions' probability functions, while tractable, are not very
useful as a basis for explicit inference procedures. An alternative is given by
the asymptotic \textit{multivariate Gaussian distribution} defined by the
\textit{multivariate central limit theorem} \citep{ash}:
\begin{equation}
\label{eqn:mclt}
  \frac{\mathbf{S}_m - m \E(\mathbf{B}_1)}{\sqrt{m}}
    \stackrel{d}{\to} N_k(\mathbf{0}, \Sigma).
\end{equation}
The limiting distribution is guaranteed to exist for all possible values of
$\mathbf{p}$, as the first two moments are bounded and therefore are always
finite.

\section{Inference on the network structure}

Let $\mathcal{U} = (\mathbf{V}, E)$ be the undirected graph underlying the DAG
$\mathcal{G} = (\mathbf{V}, A)$, defined as its unique biorientation
\citep{digraphs}. Each edge $e \in E$ of $\mathcal{U}$ corresponds to the directed
arcs in $A$ with the same incident nodes, and has only two possible states (it's
either present in or absent from the graph).

Then $e_i$, $i = 1, \ldots, |\mathbf{V}|(|\mathbf{V}| - 1)/2$ is naturally distributed
as a Bernoulli random variable
\begin{equation}
  E_i = \left\{
    \begin{aligned}
     e_i &\in E&     &\text{with probability $p_i$}& \\
     e_i &\not\in E& &\text{with probability $1 - p_i$}&
    \end{aligned}
    \right.
\end{equation}
and every set $W \subseteq \mathbf{V} \times \mathbf{V}$ (including $E$) is
distributed as a multivariate Bernoulli random variable $\mathbf{W}$ and
identified by the parameter collection
\begin{equation}
  \mathbf{p}_W = \left\{ p_w : w \subseteq W, w \neq \varnothing \right\}.
\end{equation}
The elements of $\mathbf{p}_W$ can be estimated via parametric or nonparametric
bootstrap as in \citet{friedman}, because they are functions of the DAGs
$\mathcal{G}_b$, $b = 1, \ldots, m$ through the underlying undirected
graphs $\mathcal{U}_b = (V, E_b)$. The resulting empirical probabilities
\begin{equation}
  \hat p_w = \frac{1}{m} \sum_{b=1}^m \mathbb{I}_{\left\{w \subseteq E_b\right\}}(\mathcal{U}_b),
\end{equation}
in particular
\begin{align}
  &\hat p_i = \frac{1}{m} \sum_{b=1}^m \mathbb{I}_{\left\{e_i \in E_b\right\}}(\mathcal{U}_b)&
  &\text{and}&
  &\hat p_{ij} = \frac{1}{m} \sum_{b=1}^m \mathbb{I}_{\left\{e_i \in E_b, e_j \in E_b\right\}}(\mathcal{U}_b),
\end{align}
can be used to obtain several descriptive measures and test statistics for
the variability of the network's structure.

\subsection{Interpretation of bootstrapped networks}

Considering the undirected graphs $\mathcal{U}_1, \ldots, \mathcal{U}_m$
instead of the corresponding directed graphs $\mathcal{G}_1, \ldots,
\mathcal{G}_m$ greatly simplifies the interpretation of bootstrap's results. In
particular the variability of the graphical structure can be summarized in three
cases according to the entropy \citep{itheory} of the set of the bootstrapped
networks:
\begin{itemize}
  \item \textit{minimum entropy}: all the networks learned from the bootstrap
    samples have the same structure, that is
    \begin{equation}
      E_1 = E_2 = \ldots = E_m = E.
    \end{equation}
    This is the best possible outcome of the simulation, because there is no
    variability in the estimated network. In this case the first two moments of
    the multivariate Bernoulli distribution are equal to
    \begin{align}
      &p_i = \left\{
        \begin{aligned}
         &1& &\text{if $e_i \in E$}     \\
         &0& &\text{otherwise}&
        \end{aligned}
        \right.&
      &\text{and}&
      &\Sigma = \mathbf{O}.
    \end{align}
  \item \textit{intermediate entropy}: several network structures are observed
    with different frequencies $m_b$, $\sum m_b = m$. The
    first two sample moments of the multivariate Bernoulli distribution are equal to
    \begin{align}
      &\hat p_i = \frac{1}{m} \sum_{b \,:\, e_i \in E_b} m_b&
      &\text{and}&
      &\hat p_{ij} = \frac{1}{m} \sum_{b \,:\, e_i \in E_b, e_j \in E_b} m_b.
    \end{align}
  \item \textit{maximum entropy}: all $2^{|\mathbf{V}|(|\mathbf{V}| - 1)/2}$ possible network structures appear with
    the same frequency, that is
    \begin{align}
      &\hat\Prob(\mathcal{U}_i) = \frac{1}{2^{|\mathbf{V}|(|\mathbf{V}|-1)/2}}& &i = 1, \ldots, 2^{|\mathbf{V}|(|\mathbf{V}| - 1)/2}.
    \end{align}
    This is the worst possible outcome because edges vary independently of each
    other and each one is present in only half of the networks (proof provided in
    appendix \ref{app:maxent}):
    \begin{align}
      \label{eqn:maxent}
      &p_i = \frac{1}{2}& &\text{and}& &\Sigma = \frac{1}{4} I_k.
    \end{align}
\end{itemize}

\subsection{Descriptive statistics of network's variability}
\label{sec:descriptive}

Several functions have been proposed in literature as univariate measures of
spread of a multivariate distribution, usually under the assumption of multivariate
normality (see for example \citet{muirhead} and \citet{bilodeau}). Three of
them in particular can be used as descriptive statistics for the multivariate
Bernoulli distribution:
\begin{itemize}
  \item the \textit{generalized variance}
    \begin{equation}
      \VAR_G(\Sigma) = \det(\Sigma).
    \end{equation}
  \item the \textit{total variance} (called \textit{total variation} in \citet{mardia})
    \begin{equation}
      \VAR_T(\Sigma) = \tr(\Sigma).
    \end{equation}
  \item the squared \textit{Frobenius matrix norm}
    \begin{equation}
      \VAR_N(\Sigma) = ||| \Sigma - \frac{k}{4}I_k|||_F^2.
    \end{equation}
\end{itemize}

Both the \textit{generalized variance} and the \textit{total variance} associate
high values of the statistic to unstable network structures, and are bounded
due to the properties of the covariance matrix $\Sigma$. For the total variance
it's easy to show that
\begin{equation}
\label{eq:vart}
  0 \leqslant \VAR_T(\Sigma) = \sum_{i=1}^k \sigma_{ii} \leqslant \frac{1}{4}k
\end{equation}
due to the bounds on the variances $\sigma_{ii}$ in equation \ref{thm:varconstr}.
The generalized variance is similarly bounded due to Hadamard's theorem on the
determinant of a non-negative definite matrix \citep{seber}:
\begin{equation}
  0 \leqslant \VAR_G(\Sigma) \leqslant \prod_{i=1}^k \sigma_{ii} \leqslant \left(\frac{1}{4}\right)^k.
\end{equation}
They reach the respective maxima in the \textit{maximum entropy} case and
are equal to zero only in the \textit{minimum entropy} case. The generalized
variance is also strictly convex (the maximum is reached only for $\Sigma =
\frac{1}{4}I_k$), but it is equal to zero if $\Sigma$ is rank deficient.
For this reason it's convenient to reduce $\Sigma$ to a smaller, full rank matrix
(let's say $\Sigma^*$) and compute $\VAR_G(\Sigma^*)$ instead of $\VAR_G(\Sigma)$.

The squared Frobenius norm on the other hand associates high values of the 
statistic to stable network structures. It can be rewritten in terms of the 
eigenvalues $\lambda_1, \ldots, \lambda_k$ of $\Sigma$ as
\begin{equation}
  \VAR_N(\Sigma) = \sum_{i=1}^k \left( \lambda_i - \frac{k}{4}\right)^2.
\end{equation}
It has a unique maximum (in the \textit{minimum entropy} case), which can be
computed as the solution of the constrained minimization problem in
$\boldsymbol{\lambda} = [\lambda_1, \ldots, \lambda_k]^T$
\begin{align}
  &\min_{\mathcal{D}} f(\boldsymbol{\lambda})
    = -\sum_{i=1}^k \left( \lambda_i - \frac{k}{4}\right)^2&
  &\text{subject to}&
  &\lambda_i \geqslant 0, \sum_{i=1}^k \lambda_i \leqslant \frac{k}{4}
\end{align}
using the extended Lagrange multipliers methods \citep{nocedal}. It also has a
single minimum in $\boldsymbol{\lambda}^* = [\frac{1}{4}, \ldots, \frac{1}{4}]$,
which is the projection of $[\frac{k}{4}, \ldots, \frac{k}{4}]$ onto the set
$\mathcal{D}$ and coincides with the \textit{maximum entropy} case. The proof
for these boundaries and the rationale behind the use of $\frac{k}{4}I_k$
instead of $\frac{1}{4}I_k$ are reported in appendix \ref{app:frob}.

The corresponding normalized statistics are:
\begin{align*}
  \overline{\VAR}_T(\Sigma) &= \frac{\VAR_T(\Sigma)}{\max_{\Sigma} \VAR_T(\Sigma)} = \frac{4 \VAR_T(\Sigma)}{k} \\
  \overline{\VAR}_G(\Sigma) &= \frac{\VAR_G(\Sigma)}{\max_{\Sigma} \VAR_G(\Sigma)} = 4^k \VAR_G(\Sigma) \\
  \overline{\VAR}_N(\Sigma) &= \frac{\max_{\Sigma} \VAR_N(\Sigma) - \VAR_N(\Sigma)}{\max_{\Sigma} \VAR_N(\Sigma) - \min_{\Sigma} \VAR_N(\Sigma)}
    = \frac{k^3 - 16\VAR_N(\Sigma)}{k(2k - 1)}.
\end{align*}
All of them vary in the $[0,1]$ interval and associate high values of the
statistic to networks whose structure display a high variability across the
bootstrap samples. Equivalently we can define
\begin{align*}
  \overline{\overline{\VAR}}_T(\Sigma) &= 1 - \overline{\VAR}_T(\Sigma) \\
  \overline{\overline{\VAR}}_G(\Sigma) &= 1 - \overline{\VAR}_G(\Sigma) \\
  \overline{\overline{\VAR}}_N(\Sigma) &= 1 - \overline{\VAR}_N(\Sigma)
\end{align*}
which associate high values of the statistic to networks with little
variability, and can be used as measures of distance from the \textit{maximum
entropy} case.

\begin{figure}[ht]
  \begin{center}
    \includegraphics[width=10cm]{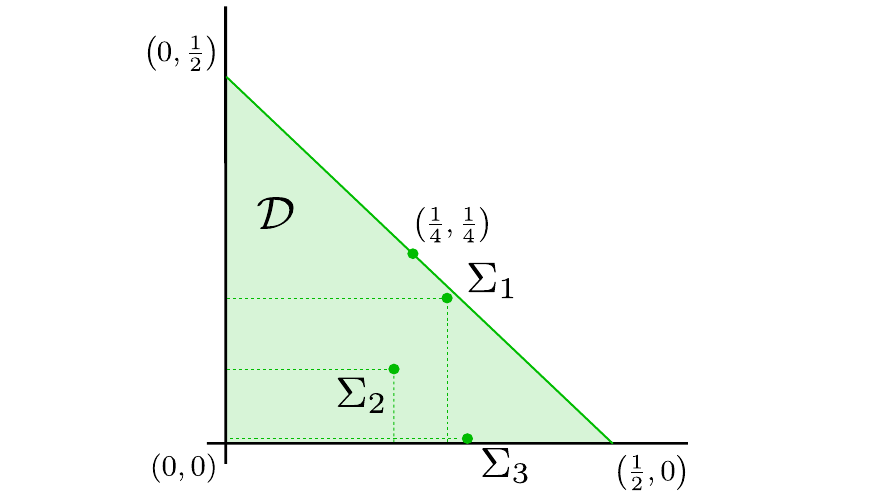}
  \end{center}
  \caption{The covariance matrices $\Sigma_1$, $\Sigma_2$ and $\Sigma_3$ represented
    as functions of their eigenvalues in $\mathcal{D}$ (green). The points $(0,0)$ 
    and $(\frac{1}{4}, \frac{1}{4})$ correspond to the \textit{minimum entropy} and
    \textit{maximum entropy} cases.}
\label{fig:base}
\end{figure}

\begin{ex}
\label{ex:base}
  Let's consider three multivariate Bernoulli distributions $\mathbf{W}_1$,
  $\mathbf{W}_2$, $\mathbf{W}_3$ with second order moments
  \begin{align*}
    &\Sigma_1 = \frac{1}{25} \begin{bmatrix} 6 & 1 \\ 1 & 6 \end{bmatrix},&
    &\Sigma_2 = \frac{1}{625} \begin{bmatrix} 66 & -21 \\ -21 & 126 \end{bmatrix},&
    &\text{and}&
    &\Sigma_3 = \frac{1}{625} \begin{bmatrix} 66 & 91 \\ 91 & 126 \end{bmatrix}.
  \end{align*}
  The eigenvalues of $\Sigma_1$, $\Sigma_2$ and $\Sigma_3$ are
  \begin{align*}
    &\boldsymbol{\lambda}_1 = \begin{bmatrix} 0.28 \\ 0.20 \end{bmatrix},&
    &\boldsymbol{\lambda}_2 = \begin{bmatrix} 0.2121 \\ 0.095 \end{bmatrix},&
    &\boldsymbol{\lambda}_3 = \begin{bmatrix} 0.3069 \\ 0.0003 \end{bmatrix}
  \end{align*}
  and the values of the generalized variance, total variance and squared Frobenius 
  matrix norm (both normalized and in the original scale) for the three covariance
  matrices are reported below.
  \begin{center}
  \begin{tabular}{|l|lll|lll|}
  \hline
             &                  &                  &                  &                             &                             & \\[-12pt]
             & $\VAR_T(\Sigma)$ & $\VAR_G(\Sigma)$ & $\VAR_N(\Sigma)$ & $\overline{\VAR}_T(\Sigma)$ & $\overline{\VAR}_G(\Sigma)$ & $\overline{\VAR}_N(\Sigma)$ \\
  \hline
  $\Sigma_1$ & $0.48$           & $0.056$          & $0.1384$         & $ 0.96$                     & $0.896$                     & $0.9642$  \\
  $\Sigma_2$ & $0.3072$         & $0.02016$        & $0.2468$         & $0.6144$                    & $0.32256$                   & $0.6752$     \\
  $\Sigma_3$ & $0.3072$     & $8.96\times 10^{-5}$ & $0.2869$         & $0.6144$                    & $0.00143$                   & $0.5682$  \\
  \hline
  \end{tabular}
  \end{center}
\end{ex}

\subsection{Asymptotic inference}
\label{sec:aymptotic}

The limiting distribution of the descriptive statistics defined above can be
derived by replacing the covariance matrix $\Sigma$ with its unbiased estimator
$\hat\Sigma$ and by considering the multivariate Gaussian distribution from equation
\ref{eqn:mclt}. The hypothesis we are interested in is
\begin{align}
\label{eqn:null}
  &H_0: \Sigma = \frac{1}{4}I_k&
  &H_1: \Sigma \neq \frac{1}{4}I_k,
\end{align}
which relates the sample covariance matrix with the one from the \textit{maximum
entropy} case.

For the total variance we have that \citep{muirhead}
\begin{equation}
  t_T = 4m \tr(\hat\Sigma) \stackrel{.}{\sim} \chi^2_{mk},
\end{equation}
and since the maximum value of $\tr(\Sigma)$ is achieved in the \textit{maximum
entropy} case, the hypothesis in \ref{eqn:null} assumes the form
\begin{align}
  &H_0: \tr(\Sigma) = \frac{k}{4}&
  &H_1: \tr(\Sigma) < \frac{k}{4}.
\end{align}
Then the observed significance value is
\begin{equation}
   \hat\alpha_T = \Prob(t_T \leqslant t_T^{oss}),
\end{equation}
and can be improved with the finite sample correction
\begin{equation}
   \tilde\alpha_T = \Prob\left(t_T \leqslant t_T^{oss} \,|\, t_T \in [0, mk]\right) = 
   \frac{\Prob(t_T \leqslant t_T^{oss})}{\Prob(t_T \leqslant mk)}
\end{equation}
which accounts for the bounds on $\VAR_T(\Sigma)$ from inequality \ref{eq:vart}.

For the generalized variance there are several possible asymptotic and approximate
distributions:
\begin{itemize}
  \item the Gaussian distribution defined in \citet{anderson}
    \begin{equation}
      t_{G_1} = \sqrt{m} \left( \frac{\det(\hat\Sigma)}{\det(\frac{1}{4}I_k)} -1 \right) \stackrel{.}{\sim} N(0, 2k).
    \end{equation}
  \item the Gamma distribution defined in \citet{steyn}
    \begin{equation}
      t_{G_2} = \frac{mk}{2} \sqrt[k]{\frac{\det(\hat\Sigma)}{\det(\frac{1}{4}I_k)}} \stackrel{.}{\sim} Ga\left(\frac{k(m+1-k)}{2}, 1\right).
    \end{equation}
  \item the saddlepoint approximation defined in \citet{saddlepoint}.
\end{itemize}

\begin{figure}[t]
  \begin{center}
    \includegraphics[width=15cm]{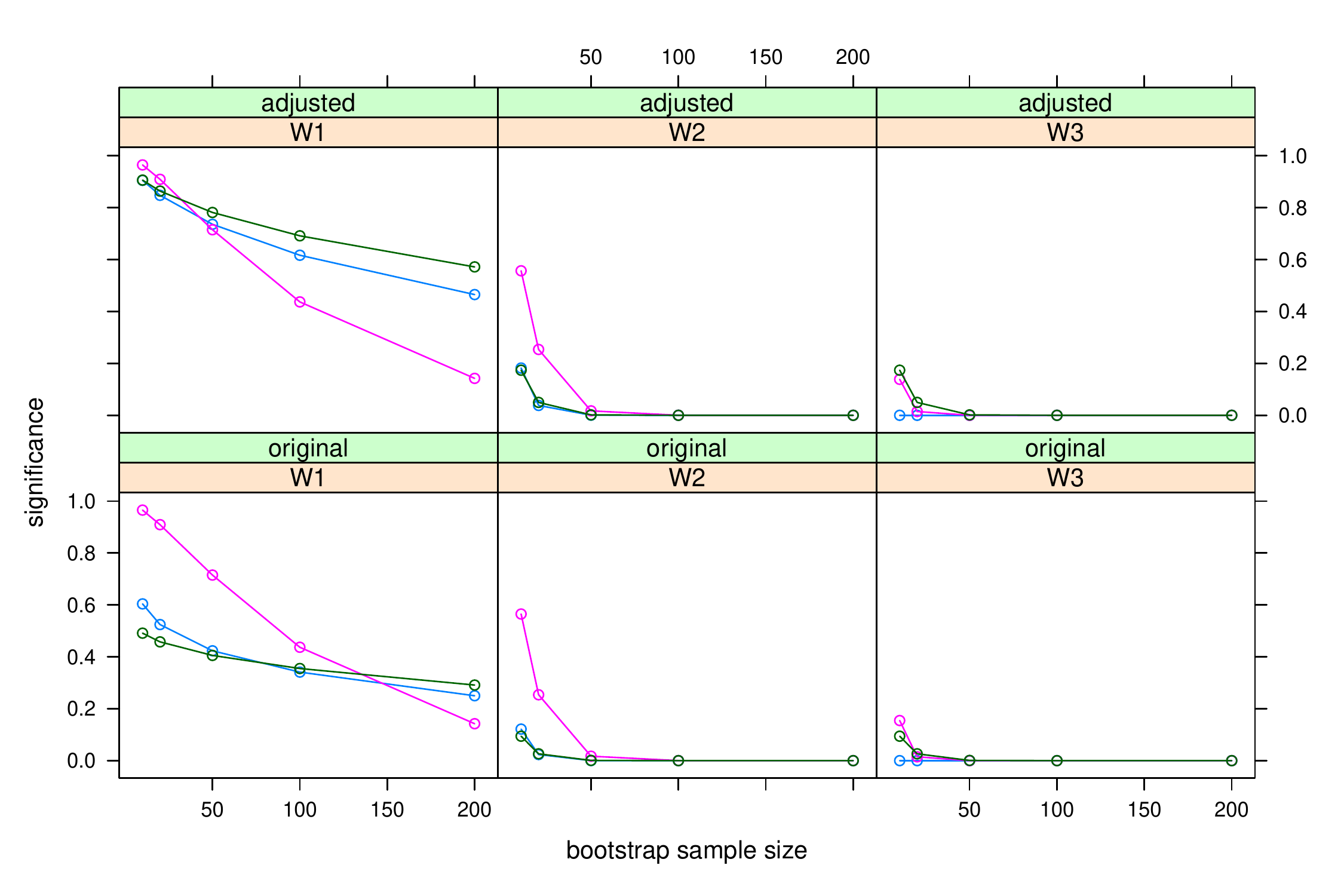}
  \end{center}
  \caption{Asymptotic significance values of $t_T$ (green), $t_{G_2}$ (blue) and 
    $t_N$ (violet) for $\Sigma_1$, $\Sigma_2$ and $\Sigma_3$ from
    table \ref{tab:asym}.}
\label{fig:asym}
\end{figure}

As before the hypothesis in \ref{eqn:null} assumes the form
\begin{align}
  &H_0: \det(\Sigma) = \det\left(\frac{1}{4}I_k\right)&
  &H_1: \det(\Sigma) < \det\left(\frac{1}{4}I_k\right).
\end{align}
The observed significance values for the Gaussian and Gamma distributions are
\begin{align}
   &\hat\alpha_{G_1} = \Prob(t_{G_1} \leqslant t_{G_1}^{oss})&
   &\hat\alpha_{G_2} = \Prob(t_{G_2} \leqslant t_{G_2}^{oss})
\end{align}
and the respective finite sample corrections for the bounds on $\det(\Sigma)$ are 
\begin{align}
   \tilde\alpha_{G_1} &= \Prob\left(t_{G_1} \leqslant t_{G_1}^{oss} \,|\, t_{G_1} \in \left[-\sqrt{m}, 0\right]\right)  
     = \frac{\Prob(t_{G_1} \leqslant t_{G_1}^{oss}) - \Prob(t_{G_1} \leqslant -\sqrt{m})}{\Prob(t_{G_1} \leqslant 0) -  \Prob(t_{G_1} \leqslant -\sqrt{m})}\\
   \tilde\alpha_{G_2} &= \Prob\left(t_{G_2} \leqslant t_{G_2}^{oss} \,|\, t_{G_2} \in \left[0, \frac{mk}{2}\right]\right) 
     =  \frac{\Prob(t_{G_2} \leqslant t_{G_2}^{oss})}{\Prob(t_{G_2} \leqslant \frac{mk}{2})}.
\end{align}

\begin{table}[t]
  \begin{center}
    \begin{tabular}{|l|lllll|}
    \hline
    \multicolumn{6}{|c|}{$t_T(\Sigma)$} \\
    \hline
               & $10$        & $20$        & $50$        & $100$          & $200$      \\
    \hline
    \multirow{2}{*}{$\Sigma_1$}
               & $0.4911379$ & $0.4576109$ & $0.4054044$ & $0.3549436$    & $0.2912432$ \\
               & $\mathbf{0.906041}$ & $\mathbf{0.863836}$ & $\mathbf{0.7814146}$ & $\mathbf{0.691495}$ & $\mathbf{0.571734}$  \\
    \multirow{2}{*}{$\Sigma_2$}
               & $0.0941934$ & $0.0263308$ & $0.0008529$ & $0.0000038$    & $1.09 \times 10^{-10}$ \\
               & $\mathbf{0.1737661}$ & $\mathbf{0.04970497}$ & $\mathbf{0.001644116}$ & $\mathbf{0.0000075}$ & $\mathbf{2.14 \times 10^{-10}}$ \\
    \multirow{2}{*}{$\Sigma_3$} 
               & $0.0941934$ & $0.0263308$ & $0.0008529$ & $0.0000038$    & $1.09 \times 10^{-10}$ \\
               & $\mathbf{0.1737661}$ & $\mathbf{0.04970497}$ & $\mathbf{0.001644116}$ & $\mathbf{0.0000075}$ & $\mathbf{2.14 \times 10^{-10}}$ \\
    \hline
    \multicolumn{6}{|c|}{$t_{G_2}(\Sigma)$} \\
    \hline
    \multirow{2}{*}{$\Sigma_1$} 
               & $0.6039442$ & $0.5242587$ & $0.4231830$     & $0.3411315$     & $0.250054$ \\
               & $\mathbf{0.9052188}$ & $\mathbf{0.8475223}$ & $\mathbf{0.7357998}$ & $\mathbf{0.6166961}$ & $\mathbf{0.4651292}$ \\
    \multirow{2}{*}{$\Sigma_2$}
               & $0.1214881$ & $0.0235145$ & $0.0002789$    & $0.0000002$     & $2.79 \times 10^{-13}$ \\
               & $\mathbf{0.1820918}$ & $\mathbf{0.03801388}$ & $\mathbf{0.000484961}$ & $\mathbf{0.00000045}$ & $\mathbf{5 \times 10^{-13}}$ \\
    \multirow{2}{*}{$\Sigma_3$}
               & $3.13 \times 10^{-10}$ & $2.03 \times 10^{-20}$ & $9.82 \times 10^{-51}$ & $4.42 \times 10^{-101}$ & $1.26 \times 10^{-201}$ \\
               & $\mathbf{4.7 \times 10^{-10}}$ & $\mathbf{3.28 \times 10^{-20}} $ & $ \mathbf{1.7 \times 10^{-50}} $ & $\mathbf{7.99 \times 10^{-101}} $ & $\mathbf{2.35 \times 10^{-201}} $\\
    \hline
    \multicolumn{6}{|c|}{$t_N(\Sigma)$} \\
    \hline
    \multirow{2}{*}{$\Sigma_1$} 
               & $0.9652055$ & $0.9091238$ & $0.7149371$     & $0.4368392$    & $0.1422717$ \\
               & $\mathbf{0.9645473}$ & $\mathbf{0.9091083}$ & $\mathbf{0.7149371}$ &  $\mathbf{0.4368392}$ & $\mathbf{0.1422717}$ \\
    \multirow{2}{*}{$\Sigma_2$} 
               & $0.5649382$ & $0.2537627$ & $0.0170906$     & $0.0001428$    & $7.48 \times 10^{-9}$ \\
               & $\mathbf{0.556708}$ & $\mathbf{0.2536360}$ & $\mathbf{0.01709067}$ & $\mathbf{0.0001428399}$ & $\mathbf{7.48 \times 10^{-9}}$ \\
    \multirow{2}{*}{$\Sigma_3$}
               & $0.1545514$ & $0.0147960$ & $0.0000085$     & $2.37 \times 10^{-11}$ & $1.34 \times 10^{-22}$ \\
               & $\mathbf{0.1385578}$ & $\mathbf{0.01462880}$ & $\mathbf{8.5 \times 10^{-06}}$ & $\mathbf{2.37 \times 10^{-11}}$ & $\mathbf{1.34 \times 10^{-22}}$ \\ 
    \hline
    \end{tabular}
  \caption{Asymptotic significance values of $t_T$, $t_{G_2}$ and $t_N$
    for $\Sigma_1$, $\Sigma_2$ and $\Sigma_3$; the ones computed with the
    finite sample corrections are reported in bold.}
  \label{tab:asym}
  \end{center}
\end{table}

The test statistic associated with the squared Frobenius norm is the test for 
the equality of two covariance matrices defined in \citet{nagao},
\begin{equation}
  t_N = \frac{m}{2} \tr\left(\left[\hat\Sigma \left(\frac{1}{4}I_k\right)^{-1} - I_k\right]^2\right)
    = \frac{m}{2} \tr\left(\left[4\hat\Sigma  - I_k\right]^2\right) \stackrel{.}{\sim} \chi^2_{\frac{1}{2}k(k+1)},
\end{equation}
because
\begin{multline}
  \tr\left(\left[4\hat\Sigma - I_k\right]^2\right) =
    \tr\left(\left[4\hat\Sigma - I_k\right] \left[4\hat\Sigma - I_k\right]\right)
    = 16 \tr\left(\left[\hat\Sigma - \frac{1}{4}I_k\right] \left[\hat\Sigma  - \frac{1}{4} I_k\right]\right) = \\
    = 16 \tr\left(\left[U \Lambda U^H - \frac{1}{4}I_k\right] \left[U \Lambda U^H  - \frac{1}{4} I_k\right]\right) = \\
    = 16 \tr\left(U\left[\Lambda - \frac{1}{4}I_k\right]U^H U \left[\Lambda - \frac{1}{4} I_k\right] U^H\right)
    = 16 \tr\left(\left[\Lambda - \frac{1}{4}I_k\right]^2\right) = \\
    = 16 \sum_{i=1}^k \left(\lambda_i - \frac{1}{4}\right)^2 = 16 ||| \hat\Sigma - \frac{1}{4}I_k|||_F^2
\end{multline}
where $U \Lambda U^H$ is the spectral decomposition of $\hat\Sigma$ (see
appendix \ref{app:frob} for and explanation of the use of $\frac{1}{4}I_k$
instead of $\frac{k}{4}I_k$). The significance value for $t_N$ is
\begin{equation}
  \hat\alpha_N = \Prob(t_N \geqslant t_N^{oss})
\end{equation}
as the hypothesis in \ref{eqn:null} becomes
\begin{align}
  &H_0: ||| \Sigma - \frac{1}{4}I_k|||_F = 0&
  &H_1: ||| \Sigma - \frac{1}{4}I_k|||_F > 0.
\end{align}
Unlike the previous statistics, Nagao's test displays a very good 
convergence speed, to the point that the finite sample correction 
for the bounds on the squared Frobenius matrix norm
\begin{equation}
   \tilde\alpha_N = \Prob\left(t_N \geqslant t_N^{oss} \,|\, t_{G_1} \in \left[0, t_N^{max}\right]\right) 
     = \frac{\Prob(t_N \geqslant t_N^{oss}) - \Prob(t_N > t_N^{max})}{\Prob(t_N \leqslant t_N^{max})}
\end{equation}
is not appreciably better than the raw significance value.

\begin{ex}
\label{ex:asymptotic}
  Let's consider again the multivariate Bernoulli distributions $\mathbf{W}_1$,
  $\mathbf{W}_2$, $\mathbf{W}_3$ and their covariance matrices $\Sigma_1$, 
  $\Sigma_2$, $\Sigma_3$ from example \ref{ex:base}. The respective asymptotic
  significance values for the statistics $t_T$, $t_{G_1}$ and $t_N$ are 
  reported in table \ref{tab:asym}.
\end{ex}

\subsection{Monte Carlo inference and parametric bootstrap}

Another approach to compute the significance values of the statistics
$\VAR_T(\Sigma)$, $\VAR_G(\Sigma)$ and $\VAR_N(\Sigma)$ is again parametric
bootstrap.

The multivariate Bernoulli distribution $\mathbf{W}_0$ specified by the hypothesis
in \ref{eqn:null} has a diagonal covariance matrix, so its components $W_{0_i}$,
$i = 1, \ldots, k$ are uncorrelated. According to theorem \ref{thm:univindep} they
are also independent, so the joint distribution of $\mathbf{W}_0$ is completely
specified by the marginal distributions $W_{0_i} \sim Ber(\frac{1}{2})$. Therefore
it's possible (and indeed quite easy) to generate observations from the null
distribution and use them to estimate the significance value of the normalized
statistics $\overline{\overline{\VAR}}_T(\Sigma)$, $\overline{\overline{\VAR}}_G(\Sigma)$
and $\overline{\overline{\VAR}}_N(\Sigma)$ defined in section \ref{sec:descriptive}:

\begin{enumerate}
  \item compute the value of test statistic $T$ on the original covariance matrix $\Sigma$.
  \item For $r = 1, 2, \ldots, R$.
  \begin{enumerate}
    \item generate $m$ sets of $k$ random samples from a $Ber(\frac{1}{2})$ distribution.
    \item compute their covariance matrix $\Sigma^*_r$.
    \item compute $T^*_r$ from $\Sigma^*_r$
  \end{enumerate}
  \item compute the Monte Carlo significance value as
  \begin{equation}
    \hat\alpha_R = \frac{1}{R} \sum_{r = 1}^R \mathbb{I}_{\{x \geqslant T\}}(T^*_r).
  \end{equation}
\end{enumerate}

This approach has two important advantages over the parametric tests defined in section 
\ref{sec:aymptotic}:
\begin{itemize}
  \item the test statistic is evaluated against the null distribution instead
    of its asymptotic approximation, thus removing any distortion caused by lack
    of convergence (which can be quite slow and problematic in high dimensions). 
  \item each simulation $r$ has a lower computational cost than the equivalent
    application of the structure learning algorithm to a bootstrap sample $b$.
    Therefore the Monte Carlo test can achieve a good precision with a smaller
    number of bootstrapped networks, allowing its application to larger problems.
\end{itemize}

\begin{figure}[t]
  \begin{center}
    \includegraphics[width=15cm]{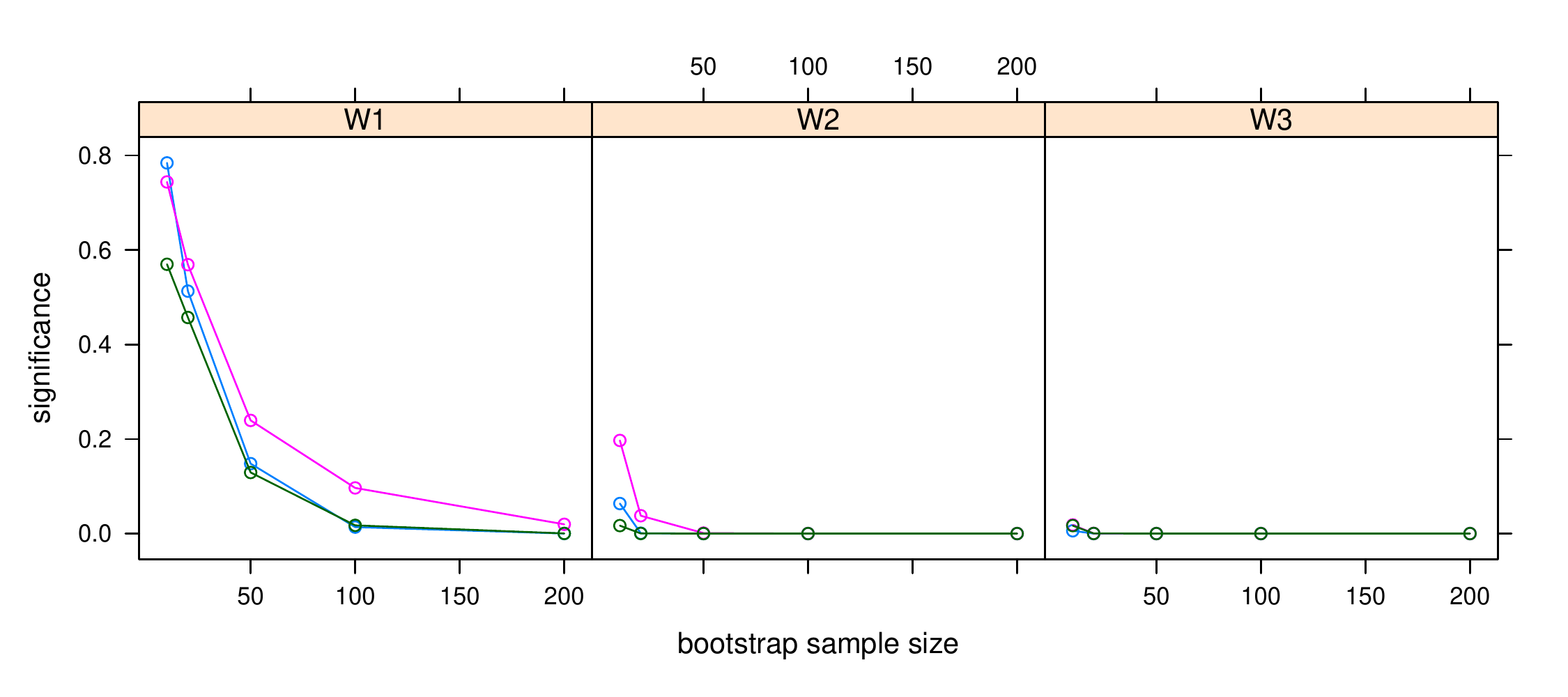}
  \end{center}
  \caption{Monte Carlo significance values for the total variance (green), the 
    generalized variance (blue) and the squared Frobenius matrix norm (violet) 
    from table \ref{tab:boot}.}
\label{fig:boot}
\end{figure}

\begin{table}[ht]
  \begin{center}
  \begin{tabular}{|l|lllll|}
  \hline
  \multicolumn{6}{|c|}{} \\[-12pt]
  \multicolumn{6}{|c|}{$\overline{\overline{\VAR}}_T(\Sigma)$} \\
  \hline
             & $10$       & $20$       & $50$       & $100$      & $200$      \\
  \hline
  $\Sigma_1$ & $0.569655$ & $0.457109$ & $0.129242$ & $0.017416$ & $0.000334$ \\
  $\Sigma_2$ & $0.016834$ & $0.000205$ & $0$        & $0$        & $0$        \\
  $\Sigma_3$ & $0.016834$ & $0.000205$ & $0$        & $0$        & $0$        \\
  \hline
  \multicolumn{6}{|c|}{} \\[-12pt]
  \multicolumn{6}{|c|}{$\overline{\overline{\VAR}}_G(\Sigma)$} \\
  \hline
  $\Sigma_1$ & $0.784102$ & $0.512839$ & $0.14788$  & $0.013678$ & $0.000094$ \\
  $\Sigma_2$ & $0.063548$ & $0.000761$ & $0$        & $0$        & $0$        \\
  $\Sigma_3$ & $0.005909$ & $0.000008$ & $0$        & $0$        & $0$        \\
  \hline
  \multicolumn{6}{|c|}{} \\[-12pt]
  \multicolumn{6}{|c|}{$\overline{\overline{\VAR}}_N(\Sigma)$} \\
  \hline
  $\Sigma_1$ & $0.743797$ & $0.568819$ & $0.239397$ & $0.096544$ & $0.019633$ \\
  $\Sigma_2$ & $0.196996$ & $0.037772$ & $0.001018$ & $0.000005$ & $0$        \\
  $\Sigma_3$ & $0.018292$ & $0.000355$ & $0$        & $0$        & $0$        \\
  \hline
  \end{tabular}
\end{center}
\caption{Bootstrap significance values from parametric bootstrap for $\Sigma_1$, $\Sigma_2$ 
  and $\Sigma_3$.}
\label{tab:boot}
\end{table}

\begin{ex}
\label{ex:bootpar}
  Let's consider the multivariate Bernoulli distributions $\mathbf{W}_1$,
  $\mathbf{W}_2$, $\mathbf{W}_3$ from examples \ref{ex:base} and \ref{ex:asymptotic}
  one last time. The corresponding significance values for the three normalized
  statistics $\overline{\overline{\VAR}}_T(\Sigma)$, $\overline{\overline{\VAR}}_G(\Sigma)$
  and $\overline{\overline{\VAR}}_N(\Sigma)$ are reported in table \ref{tab:boot} for various sizes
  of the bootstrap samples $(m = 10, 20, 50, 100, 200)$. Each one have been computed
  from $R = 10^6$ covariance matrices generated from the null distribution.
  The code used for the simulation is reported in appendix \ref{code:boot}.
\end{ex}

\section{Conclusions}

In this paper we derived the properties of several measures of variability for
the structure of a Bayesian network through its underlying undirected graph, 
which is assumed to have a multivariate Bernoulli distribution. Descriptive 
statistics, asymptotic and Monte Carlo tests were developed along with their 
fundamental properties. They can be used to compare the performance of different
learning algorithms and to measure the strength of any arbitrary subset of arcs.

\section*{Acknowledgements}

Many thanks to Adriana Brogini, my Supervisor at the Ph.D. School in Statistical
Sciences (University of Padova), for proofreading this article and giving many
useful comments and suggestions. I would also like to thank 
Giovanni Andreatta and Luigi Salce (Full Professors at the Department of Pure 
and Applied Mathematics, University of Padova) for their help in the development 
of the constrained optimization and matrix norm applications respectively.

\section*{Appendix}

\appendix

\section{Bounds on the squared Frobenius matrix norm}
\label{app:frob}

The squared Frobenius matrix norm of the difference between the covariance
matrix $\Sigma$ and the \textit{maximum entropy} matrix $\frac{1}{4}I_k$ is
\begin{equation}
   |||\Sigma - \frac{1}{4}I_k|||_F^2
    = \sum_{i=1}^k \left( \lambda_i - \frac{1}{4}\right)^2.
\end{equation}
Its unique global minimum is
\begin{equation}
   |||\Sigma - \frac{1}{4}I_k|||_F^2 = 0
\end{equation}
for $\Sigma = \frac{1}{4}I_k$ due to the fundamental properties of the matrix norms \citep{salce}.
However, it has a varying number of global maxima depending on the
dimension $k$ of $\Sigma$. They are the solutions of the constrained
minimization problem
\begin{align}
  &\min_{\mathcal{D}} f(\boldsymbol{\lambda})
    = -\sum_{i=1}^k \left( \lambda_i - \frac{k}{4}\right)^2&
  &\text{subject to}&
  &\lambda_i \geqslant 0, \sum_{i=1}^k \lambda_i \leqslant \frac{k}{4}
\end{align}
and can be computed from the Lagrangian equation and its derivatives
\begin{align}
  \mathcal{L}(\boldsymbol{\lambda}, \mathbf{s}) &= - \sum_{i=1}^k \left( \lambda_i
    - \frac{1}{4}\right)^2 - \sum_{i=1}^k s_i \lambda_i
    - s_{k+1} \left( \frac{k}{4} -\sum_{i=1}^k \lambda_i \right) \\
  \frac{\delta}{\delta \lambda_i}\mathcal{L}(\boldsymbol{\lambda}, \mathbf{s}) &=
    -2 \lambda_i + \frac{1}{2} - s_i + s_{k+1} \\
  \frac{\delta^2}{\delta^2 \lambda_i}\mathcal{L}(\boldsymbol{\lambda}, \mathbf{s}) &= -2, \qquad
  \frac{\delta^2}{\delta \lambda_i \delta \lambda_j}\mathcal{L}(\boldsymbol{\lambda}, \mathbf{s}) = 0
\end{align}
where $\mathbf{s} = [s_1. \ldots, s_{k+1}]^T$ are the Lagrangian multipliers.
This configuration of stationary points does not influence the results based on 
the asymptotic distribution of the multivariate Bernoulli distribution, but 
prevents any direct interpretation of quantities based on this difference in 
matrix norm as descriptive statistics.
\begin{figure}[t]
  \begin{center}
    \includegraphics{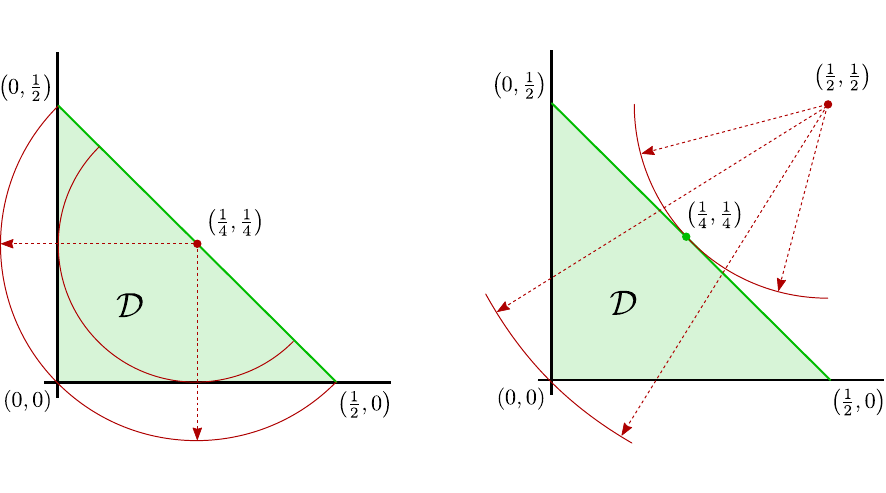}
  \end{center}
  \caption{Squared Frobenius matrix norms from $\frac{1}{4}I_K$ (on the left) and
    $\frac{k}{4}I_k$ (on the right) in $\mathcal{D}$ for $k = 2$. The green area 
    is the set $\mathcal{D}$ of the possible eigenvalues of $\Sigma$ and the red
    lines are level curves.}
\label{fig:frobenius}
\end{figure}

On the other hand the difference in squared Frobenius norm
\begin{equation}
  \VAR_N(\Sigma) = |||\Sigma - \frac{k}{4}I_k|||_F^2
    = \sum_{i=1}^k \left( \lambda_i - \frac{k}{4}\right)^2
\end{equation}
has both a unique global minimum (because it's a convex function)
\begin{equation}
  \min_{\mathcal{D}} \VAR_N(\Sigma) = \VAR_N\left(\frac{1}{4}I_k\right) 
    = \sum_{i=1}^k \left( \frac{1}{4} - \frac{k}{4}\right)^2 = \frac{k(k-1)^2}{16}
\end{equation}
and a unique global maximum
\begin{equation}
  \max_{\mathcal{D}} \VAR_N(\Sigma) = \VAR_N(\mathbf{O}) 
     = \sum_{i=1}^k \left( \frac{k}{4}\right)^2 = \frac{k^3}{16}
\end{equation}
which correspond to the \textit{minimum entropy} ($\boldsymbol{\lambda} = [0, \ldots, 0]$) 
and the \textit{maximum entropy} ($\boldsymbol{\lambda} = [\frac{1}{4}, \ldots, \frac{1}{4}]$)
covariance matrices respectively (see figure \ref{fig:frobenius}). However since 
$\frac{k}{4}I_k$ is not a valid covariance matrix for a multivariate Bernoulli 
distribution, $\VAR_N(\Sigma)$ cannot be used to derive any probabilistic result.

\section{Multivariate Bernoulli and the maximum entropy case}
\label{app:maxent}

Let's first state a simple theorem on the probability of one and two edges
in the \textit{maximum entropy} case.

\begin{thm}
  Let $\mathcal{U}_1, \ldots, \mathcal{U}_n$, $n = 2^m$, $m = |\mathbf{V}| (|\mathbf{V}| - 1)/2$ be
  all possible undirected graphs with vertex set $\mathbf{V}$ and let 
  \begin{align}
    &\Prob(\mathcal{U}_k) = \frac{1}{n}& &k = 1, \ldots, n.
  \end{align}
  Let $e_i$ and $e_j$, $i \neq j$ be two edges in $\mathbf{V} \times \mathbf{V}$. Then
  \begin{align}
    &\Prob(e_i) = \frac{1}{2}& &\text{and}& &\Prob(e_i, e_j) = \frac{1}{4}.
  \end{align}
\end{thm}
\begin{proof}
  The number of possible configurations of an undirected graph is given by
  the Cartesian product of the possible states of its $m$ edges, resulting in
  \begin{equation}
    |\{0,1\} \times \ldots \times \{0,1\}| = \left|\{0,1\}^m\right| = 2^m
  \end{equation}  
  possible undirected graphs. Then edge $e_i$ is present in
  \begin{equation}
    |\{0,1\} \times \ldots \times 1 \times \ldots \times \{0,1\}| = \left|1 \times \{0,1\}^{m-1}\right| = 2^{m-1}
  \end{equation}
  graphs and $e_i$ and $e_j$  are simultaneously present in
  \begin{equation}
    |\{0,1\} \times \ldots \times 1 \times 1 \times \ldots \times \{0,1\}| = \left|1^2 \times \{0,1\}^{m-2}\right| = 2^{m-2}
  \end{equation}
  graphs. Therefore 
  \begin{align}
    &\Prob(e_i) = \frac{2^{m-1} \Prob(\mathcal{U}_k) }{2^m \Prob(\mathcal{U}_k) } = \frac{1}{2}&
    &\text{and}&
    &\Prob(e_i, e_j) = \frac{2^{m-2} \Prob(\mathcal{U}_k) }{2^m \Prob(\mathcal{U}_k) } = \frac{1}{4}.
  \end{align}
\end{proof}

Then the values of $p_i$ and $\Sigma = [\sigma_{ij}]$ in equation 
\ref{eqn:maxent} are indeed: 
\begin{align}
  \E(e_i) &= p_i = \frac{1}{2} \\
  \VAR(e_i) &= \sigma_{ii} = p_i - p_i^2 = \frac{1}{2} - \frac{1}{4} = \frac{1}{4}\\
  \COV(e_i, e_j) &= \sigma_{ij} = p_{ij} - p_i p_j = \frac{1}{4} - \frac{1}{2} \cdot \frac{1}{2} = 0.
\end{align}
The fact that $\sigma_{ij} = 0$ for every $i \neq j$ also proves 
that the edges are independent according to theorem \ref{thm:univindep}.

\section{R code for the parametric bootstrap simulation}
\label{code:boot}

The following R function has been used to compute the significance values in
example \ref{ex:bootpar}.

{\scriptsize
\begin{verbatim}
biv.ber.sim = function(sigma, B, R, test) {

  if (test == "vart")
    FUN = function(lambda) 1/2 - sum(lambda)
  else if (test == "varg")
    FUN = function(lambda) 1/16 - prod(lambda)
  else if (test == "varn")
    FUN = function(lambda) sum((lambda - 1/4)^2)

  sim = matrix(0L, nrow = B, ncol = 2)
  tstar = numeric(R)

  s0 = eigen(sigma)$values
  t0 = FUN(s0)

  for (i in 1:R) {

    for (j in 1:B)
      sim[j, ] = rbinom(2, 1, 1/2)

    p = prop.table(table(factor(sim[, 1], levels = c(0,1)),
          factor(sim[, 2], levels = c(0,1))))

    sigmastar = matrix(
      c(sum(p[2,]) * (1 - sum(p[2,])),
        p[2,2] - sum(p[,2])*sum(p[2,]),
        p[2,2] - sum(p[,2])*sum(p[2,]),
        sum(p[,2]) * (1 - sum(p[,2]))),
      nrow = 2, ncol = 2, byrow = TRUE)

    sstar = eigen(sigmastar)$values
    tstar[i] = FUN(sstar)

  }

  return(length(tstar[tstar >= t0])/R)

}
\end{verbatim}
}

The three covariance matrices $\Sigma_1$, $\Sigma_2$ and $\Sigma_3$ have been
created in R with the following commands.
{\scriptsize
\begin{verbatim}
sigma1 = matrix(c(6, 1, 1, 6)/25, ncol = 2)
sigma2 = matrix(c(66, -21, -21, 126)/625, ncol = 2)
sigma3 = matrix(c(66, 91, 91, 126)/625, ncol = 2)
\end{verbatim}
}

All the simulations have been performed on a Core Duo 2 machine with 1GB of RAM,
with R 2.9.0 \citep{r} and an updated Debian GNU/Linux distribution.

\section{R code for the asymptotic inference}
{\scriptsize
\begin{verbatim}

total.variance = function(sigma, b, adjusted = FALSE) {

  res = pchisq(4 * b * sum(diag(sigma)), 2 * b, lower.tail = TRUE)

  if (adjusted)
    res = res / pchisq(2 * b, 2 * b, lower.tail = TRUE)

  return(res)

}

generalized.variance = function(sigma, b, adjusted = FALSE) {

  res = pgamma(4 * b * sqrt(det(sigma)), b - 1, 1, lower.tail = TRUE)

  if (adjusted)
    res = res / pgamma(b, b - 1, 1, lower.tail = TRUE)

  return(res)

}

frobenius.norm = function(sigma, b, adjusted = FALSE) {

  res = pchisq(8 * b * sum((eigen(sigma)$values - 1/4 )^2), 3, lower.tail = FALSE)

  if (adjusted)
    res = (res - pchisq(b, 3, lower.tail = FALSE)) / pchisq(b, 3, lower.tail = TRUE)

  return(res)

}
\end{verbatim}
}

\pagebreak


\begin{thebibliography}{31}
\newcommand{\enquote}[1]{``#1''}
\providecommand{\natexlab}[1]{#1}
\providecommand{\url}[1]{\texttt{#1}}
\providecommand{\urlprefix}{URL }
\expandafter\ifx\csname urlstyle\endcsname\relax
  \providecommand{\doi}[1]{doi:\discretionary{}{}{}#1}\else
  \providecommand{\doi}{doi:\discretionary{}{}{}\begingroup
  \urlstyle{rm}\Url}\fi
\providecommand{\eprint}[2][]{\url{#2}}

\bibitem[{Anderson(2003)}]{anderson}
Anderson TW (2003).
\newblock \emph{An Introduction to Multivariate Statistical Analysis}.
\newblock Wiley-Interscience, 3rd edition.

\bibitem[{Ash(2000)}]{ash}
Ash RB (2000).
\newblock \emph{Probability and Measure Theory}.
\newblock Academic Press, 2nd edition.

\bibitem[{Bang-Jensen and Gutin(2009)}]{digraphs}
Bang-Jensen J, Gutin G (2009).
\newblock \emph{Digraphs: Theory, Algorithms and Applications}.
\newblock Springer, 2nd edition.

\bibitem[{Bilodeau and Brenner(1999)}]{bilodeau}
Bilodeau M, Brenner D (1999).
\newblock \emph{Theory of Multivariate Statistics}.
\newblock Springer-Verlag.

\bibitem[{Butler \emph{et~al.}(1992)Butler, Huzurbazar, and
  Booth}]{saddlepoint}
Butler RW, Huzurbazar S, Booth JG (1992).
\newblock \enquote{Saddlepoint Approximations for the Generalized Variance and
  Wilks' Statistic.}
\newblock \emph{Biometrika}, \textbf{79}(1), 157 -- 169.

\bibitem[{Chickering(2002)}]{ges}
Chickering DM (2002).
\newblock \enquote{Optimal Structure Identification with Greedy Search.}
\newblock \emph{Journal of Machine Learning Research}, \textbf{3}, 507--554.

\bibitem[{Cover and Thomas(2006)}]{itheory}
Cover TA, Thomas JA (2006).
\newblock \emph{Elements of Information Theory}.
\newblock Wiley.

\bibitem[{Efron and Tibshirani(1993)}]{efron}
Efron B, Tibshirani R (1993).
\newblock \emph{An Introduction to the Bootstrap}.
\newblock Chapman \& Hall.

\bibitem[{Elidan(2001)}]{bnr}
Elidan G (2001).
\newblock \enquote{Bayesian Network Repository.}
\newblock \urlprefix\url{http://www.cs.huji.ac.il/labs/compbio/Repository}.

\bibitem[{Friedman \emph{et~al.}(1999)Friedman, Goldszmidt, and
  Wyner}]{friedman}
Friedman N, Goldszmidt M, Wyner A (1999).
\newblock \enquote{Data Analysis with Bayesian Networks: A Bootstrap Approach.}
\newblock In \enquote{Proceedings of the 15th Annual Conference on Uncertainty
  in Artificial Intelligence (UAI-99),} pp. 206 -- 215. Morgan Kaufmann.

\bibitem[{Friedman \emph{et~al.}(2000)Friedman, Linial, and Nachman}]{gene}
Friedman N, Linial M, Nachman I (2000).
\newblock \enquote{Using Bayesian Networks to Analyze Expression Data.}
\newblock \emph{Journal of Computational Biology}, \textbf{7}, 601--620.

\bibitem[{Holmes and Jain(2008)}]{holmes}
Holmes DE, Jain LC (eds.) (2008).
\newblock \emph{Innovations in Bayesian Networks: Theory and Applications},
  volume 156 of \emph{Studies in Computational Intelligence}.
\newblock Springer.

\bibitem[{Imoto \emph{et~al.}(2002)Imoto, Kim, Shimodaira, Aburatani, Tashiro,
  Kuhara, and Miyano}]{imoto}
Imoto S, Kim SY, Shimodaira H, Aburatani S, Tashiro K, Kuhara S, Miyano S
  (2002).
\newblock \enquote{Bootstrap Analysis of Gene Networks Based on Bayesian
  Networks and Nonparametric Regression.}
\newblock \emph{Genome Informatics}, \textbf{13}, 369--370.

\bibitem[{Jungnickel(2008)}]{graphs}
Jungnickel D (2008).
\newblock \emph{Graphs, Networks and Algorithms}.
\newblock Springer, 3rd edition.

\bibitem[{Korb and Nicholson(2004)}]{korb}
Korb K, Nicholson A (2004).
\newblock \emph{Bayesian Artificial Intelligence}.
\newblock Chapman and Hall.

\bibitem[{Krummenauer(1998{\natexlab{a}})}]{mvbsim}
Krummenauer F (1998{\natexlab{a}}).
\newblock \enquote{Efficient Simulation of Multivariate Binomial and Poisson
  Distributions.}
\newblock \emph{Biometrical Journal}, \textbf{40}(7), 823 -- 832.

\bibitem[{Krummenauer(1998{\natexlab{b}})}]{krummenauer}
Krummenauer F (1998{\natexlab{b}}).
\newblock \enquote{Limit theorems for multivariate discrete distributions.}
\newblock \emph{Metrika}, \textbf{47}(1), 47 -- 69.

\bibitem[{Larra{\~n}aga \emph{et~al.}(1997)Larra{\~n}aga, Sierra, Gallego,
  Michelena, and Picaza}]{larranaga}
Larra{\~n}aga P, Sierra B, Gallego MJ, Michelena MJ, Picaza JM (1997).
\newblock \enquote{Learning Bayesian Networks by Genetic Algorithms: A Case
  Study in the Prediction of Survival in Malignant Skin Melanoma.}
\newblock In \enquote{Proceedings of the 6th Conference on Artificial
  Intelligence in Medicine in Europe (AIME'97),} pp. 261 -- 272.

\bibitem[{Lo\`eve(1977)}]{loeve}
Lo\`eve M (1977).
\newblock \emph{Probability Theory}.
\newblock Springer-Verlag, 4th edition.

\bibitem[{Mardia \emph{et~al.}(1979)Mardia, Kent, and Bibby}]{mardia}
Mardia KV, Kent JT, Bibby JM (1979).
\newblock \emph{Multivariate Analysis}.
\newblock Academic Press.

\bibitem[{Margaritis(2003)}]{mphd}
Margaritis D (2003).
\newblock \emph{Learning Bayesian Network Model Structure from Data}.
\newblock Ph.D. thesis, School of Computer Science, Carnegie-Mellon University,
  Pittsburgh, PA.
\newblock Available as Technical Report CMU-CS-03-153.

\bibitem[{Muirhead(1982)}]{muirhead}
Muirhead RJ (1982).
\newblock \emph{Aspects of Multivariate Statistical Theory}.
\newblock Wiley-Interscience.

\bibitem[{Nagao(1973)}]{nagao}
Nagao H (1973).
\newblock \enquote{On Some Test Criteria for Covariance Matrix.}
\newblock \emph{The Annals of Statistics}, \textbf{1}(4), 700 -- 709.

\bibitem[{Nocedal and Wright(1999)}]{nocedal}
Nocedal J, Wright SJ (1999).
\newblock \emph{Numerical Optimization}.
\newblock Springer-Verlag.

\bibitem[{{R Development Core Team}(2009)}]{r}
{R Development Core Team} (2009).
\newblock \emph{R: A Language and Environment for Statistical Computing}.
\newblock R Foundation for Statistical Computing, Vienna, Austria.
\newblock \urlprefix\url{http://www.r-project.org}.

\bibitem[{Salce(1993)}]{salce}
Salce L (1993).
\newblock \emph{Lezioni sulle matrici}.
\newblock Zanichelli.

\bibitem[{Seber(2008)}]{seber}
Seber GAF (2008).
\newblock \emph{A Matrix Handbook for Stasticians}.
\newblock Wiley-Interscience.

\bibitem[{Spirtes \emph{et~al.}(2001)Spirtes, Glymour, and Scheines}]{spirtes}
Spirtes P, Glymour C, Scheines R (2001).
\newblock \emph{Causation, Prediction and Search}.
\newblock MIT Press.

\bibitem[{Steyn(1978)}]{steyn}
Steyn HS (1978).
\newblock \enquote{On Approximations for the Central and Noncentral
  Distribution of the Generalized Variance.}
\newblock \emph{Journal of the American Statistical Association},
  \textbf{73}(363), 670 -- 675.

\bibitem[{Tsamardinos \emph{et~al.}(2003)Tsamardinos, Aliferis, and
  Statnikov}]{iamb}
Tsamardinos I, Aliferis CF, Statnikov A (2003).
\newblock \enquote{Algorithms for Large Scale Markov Blanket Discovery.}
\newblock In \enquote{Proceedings of the Sixteenth International Florida
  Artificial Intelligence Research Society Conference,} pp. 376--381. AAAI
  Press.

\bibitem[{Yaramakala and Margaritis(2005)}]{fastiamb}
Yaramakala S, Margaritis D (2005).
\newblock \enquote{Speculative Markov Blanket Discovery for Optimal Feature
  Selection.}
\newblock In \enquote{ICDM '05: Proceedings of the Fifth IEEE International
  Conference on Data Mining,} pp. 809--812. IEEE Computer Society, Washington,
  DC, USA.

\end{thebibliography}
\end{document}